\documentclass[10pt,reprint,floatfix,aps,pra,superscriptaddress]{revtex4-1}
\usepackage[utf8]{inputenc}
\usepackage{times,textcomp}
\usepackage{amsfonts,amssymb,amsmath,amsthm}
\usepackage{graphicx,epsfig,wrapfig,float}
\usepackage[inline]{enumitem}
\usepackage{color}
\usepackage{booktabs,dcolumn}
\usepackage[colorlinks=false,citecolor=blue,urlcolor=blue,pdfencoding=auto,psdextra]{hyperref}
\usepackage{natbib}

\newtheorem{lemma}{\textbf{Lemma}}
\newcommand{\be}{\begin{eqnarray}}
\newcommand{\ee}{\end{eqnarray}}
\newcommand{\ket}[1]{\left | #1 \right\rangle}
\newcommand{\bra}[1]{\left \langle #1 \right |}

\renewcommand{\epsilon}{\varepsilon}
\begin{document}
\title{Device independent witness of arbitrary dimensional quantum systems employing binary outcome measurements}
\author{Miko{\l}aj Czechlewski}
\email{mczechlewski@inf.ug.edu.pl}
\affiliation{Institute of Informatics, National Quantum Information Centre, Faculty of Mathematics, Physics and Informatics, University of Gda\'nsk, Wita Stwosza 57, 80-308 Gda\'nsk, Poland}
\author{Debashis Saha}
\email{saha@cft.edu.pl}
\affiliation{Institute of Theoretical Physics and Astrophysics, National Quantum Information Centre, Faculty of Mathematics, Physics and Informatics, University of Gda\'nsk, Wita Stwosza 57, 80-308 Gda\'nsk, Poland}
\affiliation{Center for Theoretical Physics, Polish Academy of Sciences, Al. Lotnik\'{o}w 32/46, 02-668 Warsaw, Poland}
\author{Armin Tavakoli}
\email{armin.tavakoli@unige.ch}
\affiliation{Département de Physique Appliquée, Université de Genève, CH-1211 Genève, Switzerland}
\author{Marcin Paw{\l}owski}
\email{dokmpa@univ.gda.pl}
\affiliation{Institute of Theoretical Physics and Astrophysics, National Quantum Information Centre, Faculty of Mathematics, Physics and Informatics, University of Gda\'nsk, Wita Stwosza 57, 80-308 Gda\'nsk, Poland}
\begin{abstract}
Device independent dimension witnesses (DW) are a remarkable way to test the dimension of a quantum system in a prepare-and-measure scenario imposing minimal assumptions on the internal features of the devices. However, as the dimension increases, the major obstacle in the realization of DW arises due to the requirement of many outcome quantum measurements. In this article, we propose a new variant of a widely studied communication task (random access code) and take its average payoff as the DW. The presented DW applies to arbitrarily large quantum systems employing only binary outcome measurements. 
\end{abstract}
\maketitle
\section{Introduction}
Realizing higher-dimensional quantum systems with full control is one of the crucial barriers towards implementing many quantum information processing protocols and testing the foundations of physics. While the process of quantum tomography allows us to reconstruct a quantum system, however, it requires the assumption of fully characterized measurement devices. The device independent framework \cite{Wehner2008,Gallego2010} in a \textit{prepare-and-measure} experiment provides a methodology to obtain a lower bound on the dimension without assuming the internal features of the devices. Moreover, quantum advantages in information processing, for example, quantum communication complexity \cite{Brassard2003,Buhrman2010} are linked to this approach. Despite its merits, implementing device independent dimension witnesses (DWs) for higher dimensional quantum systems  \cite{Cabello2012,Hendrych2012,DAmbrosio2014,Ahrens2014,Aguilar2017_1} faces several complications. 

One of the problems in many existing protocols is the requirement of $d$ outcome measurements. As the dimension increases, performing many outcome measurements \cite{Tavakoli2017} becomes practically difficult due to the facts that,
\begin{enumerate*}[label=\alph*)]
\item measurement outcomes turn coarse-grained,
\item the system becomes more prone to decoherence.
\end{enumerate*}
In some cases, one may impose additional assumptions, for instance simulating $d$ outcome measurements by many binary outcome measurements. However, this approach fails to fulfill the requirements of DWs in the strict sense.

Another difficulty arises from the fact that the number of different preparations and measurements (i.e.\, the total number of inputs in the devices) also increases as one seeks to certify higher dimensional system. As a result, the experimental errors grow large due to the finite number of trials and imperfections in the experiment.
 
Furthermore, the applicability of a desired figure of merit, used as DW, should not be limited by a particular dimension. Rather, it should be applicable to test systems of an arbitrarily high dimension. 

In this article, we overcome these challenges by proposing a class of DWs based on random access codes \cite{Ambainis2008} for quantum systems of an arbitrary dimension. In the simplest scenario, a DW can be interpreted as a task carried out by two parties. In each run of the task, the sender Alice obtains an input in the form of a classical variable $a$ and communicates a system to the receiver Bob. Apart from the communicated message, Bob also receives an input $y$ and produces an output $b$. The figure of merit, denoted by $\overline{T}$, of the task could be an arbitrary linear function of the statistics $\overline{T}=\sum_{a,y,b} p(a,y) T(a,y,b) p(b|a,y)$, where $p(b|a,y)$\footnote{$p(a,y)$ could be absorbed into $T(a,y,b)$. Nevertheless, the stated form provides a simple intuition.} refers to the probability of obtaining the output $b$ given the inputs $a,y$, and $T(a,y,b)$ denotes the payoff to that event. Assuming the dimension of the communicated system is $d$, one can obtain the optimal value of the figure of merit, denoted by $\overline{T}^{c}$, for a classical implementation. 
Obtaining a value greater than $\overline{T}^{c}$ from the observed statistics certifies the communicated quantum system to be of at least dimension $d$. Quantum random access codes (QRACs), a primitive quantum communication protocol \cite{Wiesner1983,Ambainis1999,Nayak1999}, can be used for this purpose. The original study of QRACs was restricted to two-dimensional systems \cite{Ambainis2008} and was later generalized to higher dimensions \cite{Galvao2002,Casaccino2008,Tavakoli2015} yielding several interesting results in quantum communication \cite{Hameedi2017,Hameedi2015,Aguilar17_2,Farkas2018}. 

There are advantages of using RAC as DWs. The upper bound on $\overline{T}^{c}$ can be obtained for any $d$. Besides, the number of inputs in the devices increases polynomially with $d$. Note that one can exploit the quantum communication complexity tasks \cite{Buhrman2010}, which involve binary outcome measurement for dimension witnessing, but in that case, the input size grows exponentially with $d$. However, the generalized RAC requires $d$ outcome measurements. To tackle this issue we have introduced a version of RAC, namely, binary RAC. This involves only binary outcome measurements and provides a method to obtain the upper bound of $\overline{T}^{c}$ applicable to arbitrary $d$. 

The paper is organized as follows: first, we describe the generalization of $d$-dimensional RAC, along with the proof of optimal classical protocols and bounds. Next, we propose the binary version (i.e.\ , the outcome $b$ is binary) of  $d$-dimensional RAC taking into account a wider class of payoff function. Then, we derive a condition on the payoff function such that the optimal classical protocol is the same as in a standard RAC. Further, we provide the classical bound and a quantum protocol that violates the proposed DW for arbitrary $d$.

\section{Standard $d$-dimensional  Random Access Code}
\begin{figure}[h]
\centering
\includegraphics[width=\linewidth]{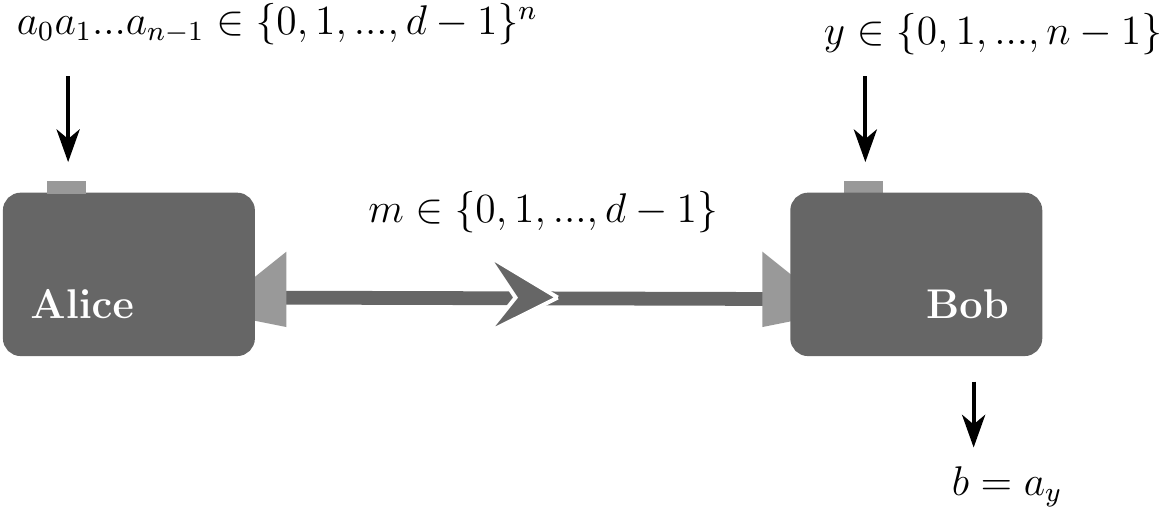}
\caption{Scheme of $d$RAC. Alice gets the input $a_{0},\dots a_{n-1}$ and sends a message $m$ to Bob. Besides the message, Bob also receives the input $y\in\{0,\dots,n-1\}$. His task is to give the output $b$, which obeys the relation $b=a_{y}$.}
\label{fig:dRAC}
\end{figure}
Standard $d$-dimensional random access codes ($d$RAC) are a natural generalization of two-dimensional random access code \cite{Ambainis2008,Tavakoli2015}. Alice receives $n$ numbers $a_{0},\dots a_{n-1}$, where $a_{i}\in \{0,\dots, d-1\}$. Then she sends a $d$-valued (one dit\footnote{By dit we mean $d$-dimensional classical system}) message $m\in\{0,\dots d-1\}$ to Bob. Bob gets an input $y\in\{0,\dots,n-1\}$. He needs to give the output $b$, which obeys the relation $b=a_{y}$ (figure \ref{fig:dRAC}). Specifically, we are interested in the average success probability in the case of the inputs $a,y$ being uniformly distributed and $T(a,y,b)=\delta_{b,a_{y}}$, 
\begin{equation}
\overline{T}_{S}=\frac{1}{nd^{n}}\sum_{a,y} p(b=a_{y}|a,y).
\end{equation}
Since the communicated message $m$ is constrained to be $d$-valued, it is evident that achieving average success probability equals to $1$ is impossible. The aim is to find an optimal strategy for the parties, which gives the largest average success probability.

Following the result in \cite{Ambainis2008} for $d=2$, it has been mentioned in \cite{Tavakoli2015} and shown later in \cite{Ambainis2015} that coding by majority and identity decoding is an optimal strategy for $d$RAC.
In the next two subsections, we demonstrate an alternative shorter proof of this fact and subsequently provide an expression of the optimal average success probability.
\subsection{Optimal classical strategy}\label{sec:OptClassStr}
Due to the linearity of the of figure of merit, it is sufficient to consider only deterministic encoding and decoding strategies to maximize the average success probability. Let us denote the dit-string $a_{0}\dots a_{n-1}$ by $a$. Any encoding strategy can be described by a function $E: \{a\} \equiv \{0,\dots,d-1\}^n \mapsto \{m\} \equiv \{0,1,\dots,d-1\}$ and the probability of sending $m$ for input $a$ is $\delta_{m,E(a)}$. While any decoding for Bob's input $y$ is described as a function $D_y: \{m\} \equiv \{0,1,\dots,d-1\} \mapsto \{b\} \equiv \{0,1,\dots,d-1\}$ and $\delta_{b,D_y(m)}$ is the probability of outputting $b$ when message $m$ is received. Thus, the classical average success probability in the standard RAC is
\begin{align}\label{eq:avrdRACgen}
\overline{T}_{S}^{c}&=\frac{1}{nd^{n}}\sum_{a,y} p(b=a_{y}|a,y) \\
&=\frac{1}{nd^{n}} \sum_m \sum_{a,y} \delta_{m,E(a)} \ \delta_{a_{y},D_{y}(m)} \nonumber \\
&=\frac{1}{nd^{n}} \sum_{m,a} \delta_{m,E(a)} \left(\sum_{y} \delta_{a_{y},D_{y}(m)} \right) \nonumber \\
&\leq\frac{1}{nd^{n}} \sum_{a} \max_{m}  \left(\sum_{y} \delta_{a_{y},D_{y}(m)} \right)\nonumber .
\end{align}
From the above expression, one can observe that for given decoding strategy $D_{y}(m)$ the optimal encoding will be the following
\begin{align}\label{eq:opten}
\delta_{m,E(a)}=1\ &\text{if}\\
\forall m' \in  \{0,\dots,d-1 \}\ 
&\sum_{y}\delta_{a_{y},D_{y}(m)}\geq\sum_{y} \delta_{a_{y},D_{y}(m')}.\nonumber
\end{align}
We can reduce the possibility of all decoding functions into two ways:
\begin{enumerate}[label=\alph*)]
\item identity decoding i.e.\ $\forall y,m, \ D_{y}(m)=m$, \label{itm:deit1}
\item not identity decoding, $\exists y,m$ such that $D(m)\neq m$. Here the mapping $D_{y}$ could be one-to-many in general.\label{itm:deit2}
\end{enumerate}
\begin{lemma}
There exists an optimal classical strategy with identity decoding \ref{itm:deit1}.
\end{lemma}
\begin{proof}
We will show that for the case described in \ref{itm:deit2}, there exists a strategy obtaining the same average success probability like for the identity decoding \ref{itm:deit1}. Let $D_{y}^{\leftarrow}(b)$ be the domain of $b$, i.e. \ , the set of $m$ such that $D_y(m)=b$. If $b$ does not exist in the range of $D_y$, we define $D_{y}^{\leftarrow}(b)=b$. We denote $D_{y}^{\leftarrow} (a)$ is the set of dit string $a'\equiv a'_{0}\dots a'_{n-1}$ such that $D_{y}(a'_{y})=a_{y}$. Thus, $D_{y}^{\leftarrow} (a)$ acts on the $y$-th dit of the dit string. If there is a classical strategy, having an encoding function $E$ and decoding functions $D_{y}$ (where $D_{y}(m) \neq m$ for some $y,m$), we can construct new encoding and decoding functions as follows
\begin{align}\label{eq:encoddec} 
E'(D_{0}^{\leftarrow}D_{1}^{\leftarrow}\dots D_{n-1}^{\leftarrow} (a))&=E(a) \\
\forall y,m, \ D'_y(m)&=m. \nonumber
\end{align} 
Now, if the strategy $(E,D_{y})$ gives the correct answer for the input $(a,y)$ then the modified strategy $(E', D_{y}{'})$ gives the correct answer for at least one of the inputs $(D_{0}^{\leftarrow}\dots D^{\leftarrow}_{n-1}(a), y)$. Thus, the average success probability for the modified strategy $(E', D')$ is equal or greater than the strategy $(E,D)$.
\end{proof}
From \eqref{eq:opten} we conclude that optimal encoding is as follows
\begin{align}\label{eq:optde}
\delta_{m,E(a)}=1 \ &\text{such that}\\  
\forall m' \in  \{0,\dots,d-1\}
&\sum_{y} \delta_{a_{y},m} \geq \sum_{y} \delta_{a_{y},m'}.\nonumber
\end{align}
In other words, the optimal strategy for Alice is to communicate the majority dit of the input string and $b=m$.

\subsection{Average success probability}
Now we calculate the classical average success probability for an $n$ dit string. The total number of possible inputs is $nd^{n}$. In the $n$ dit string, which is given to Alice, the $i$-th dit $(i \in \{0,1,2,\dots,(n-1)\})$ appears $n_{i}$ times in the string $a$. The number of ways it may occur is the same as the number of solutions in non-negative integers of the equation 
\begin{align}\label{eq:nways}
n_{0}+n_{1}+n_{2}+\dots+n_{d-1}&=n.
\end{align}
The above equation \eqref{eq:nways} is a special case of the equation
\begin{align}\label{eq:DioEq}
c_{0}n_{0}+c_{1}n_{1}+c_{2}n_{2}+\dots+c_{d-1}n_{d-1}&=n,
\end{align}
with  all coefficients $\{c_{0},c_{1},c_{2},\dots ,c_{d-1}\}$ equal 1.
The equation \eqref{eq:DioEq} is known in number theory as the Diophantine equation of Frobenius and it is connected with the Frobenius coin problem and  Frobenius' number \cite{Erdos1972,Dixmier1990}. The total number of possible solutions of \eqref{eq:nways} is $\binom{n+d-1}{d-1}$ \cite{Lint2001book}. For each solution Alice will communicate $\max\{n_{0},n_{1},\dots,n_{d-1}\}$ to Bob. So the number of successful inputs is given by $\frac{n!}{n_{0}!n_{1}!\dots n_{d-1}!}\max\{n_{0},n_{1},\dots,n_{d-1}\}$, as $\frac{n!}{n_{0}!n_{1}!\dots n_{d-1}!}$ is the number of possible combinations for an $n$ dit string with a given set of $n_{i}$'s, and $\max\{n_{0},n_{1},\dots,n_{d-1}\}$ is the number of times where Bob will guess the correct dit. Therefore, the average success probability is given by
\begin{align}\label{eq:avrdRAC} 
\overline{T}_{S}^{c}&=\frac{1}{nd^n}\sum \frac{n!}{n_{0}!n_{1}!\dots n_{d-1}!}\max\{n_{0},n_{1},\dots,n_{d-1}\}, 
\end{align} 
where the summation is over all $\binom{n+d-1}{d-1}$ possible solutions of \eqref{eq:nways}.

\section{Binary Random Access Code}
A binary random access code (figure \ref{fig:bRAC}) is a communication complexity problem based on the standard $d$RAC. Two parties, Alice and Bob, are given the following task: Alice receives $n$ dits $a=a_{0},\dots a_{n-1}$ same as in the standard $d$RAC. She sends a $d$-valued message to Bob. However, Bob gets two inputs $y\in \{0,1,\dots,n-1\}$ and $k\in \{0,1,\dots,d-1\}$. He needs to answer the question: is $a_{y}=k$? Bob encodes his answer in a variable $G$ which is 0 when his guess is YES and 1 for NO.
\begin{figure}[t]
\centering
\includegraphics[width=\linewidth]{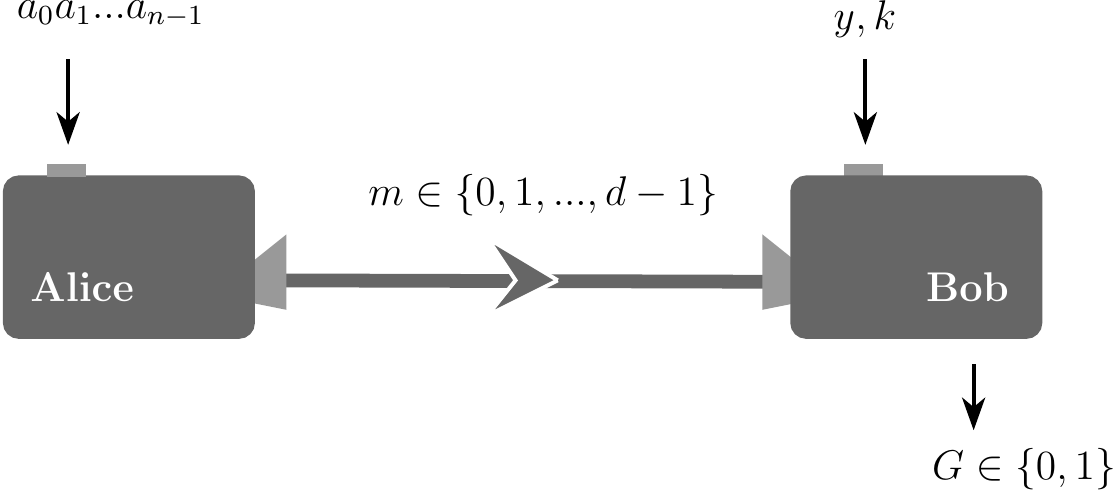}
\caption{Scheme of binary RAC (BRAC). Alice gets the input $a_{0},\dots a_{n-1}$ and sends the message $m$ to Bob. Besides the message Bob receives two inputs $y\in \{0,1,\dots,n-1\}$ and $k\in \{0,1,\dots,d-1\}$. His task is to guess whether $a_{y}=k$ or not. His answer is encoded in $G$, which is 0 when his guess is YES and 1, when it is NO.}
\label{fig:bRAC}
\end{figure}
\subsection{Defining average payoff function}
We are free to reward the parties with any number of points, specified by a payoff function $T(a,y,k,G)$. Therefore, for simplicity we assume that this function does not depend on the values of numbers $a_{i}$ in the input $a$ with indices different than $y$. Hence, we assign $T$ only two values 
\begin{align}\label{eq:assPoints}
T(a_y,k,G)&=\left\{ 
\begin{array}{ccccc}
T_{YES}&\text{when}&G=0\ \text{and}&a_{y}=k \\
1 & \text{when} &G=1\ \text{and}&a_{y}\neq k.\\
\end{array} \right.
\end{align}
We are interested in the average payoff function, which is a linear combination of payoffs for all possible uniformly distributed inputs. Without loss of generality, we can normalize average payoff such that it takes the value within $[0,1]$. Thus, for binary RAC with payoffs defined in \eqref{eq:assPoints} we have
\begin{align}\label{eq:avrBrac}
\overline{T}_{B}&=\frac{1}{nd^{n}T_{d}}\bigg[\sum\limits_{a,y,k}\Big(p(G=0|a,y,k,a_{y}=k)T_{YES}+\\
&+p(G=1|a,y,k,a_{y}\neq k)\Big)\bigg]\ ,\nonumber
\end{align}
where $T_{d}=T_{YES}+d-1$  such that $\overline{T}_B$ is normalized.

\subsection{Optimal classical strategy for Bob}

For finding the optimal classical strategy for Bob, first we split him into two parts $B_{I}$ (initial Bob) and $B_{F}$ (final Bob). $B_{I}$ gets the message $m$ from Alice, receives input $y$ and forwards $d$ long bit string $b=b_{0},\dots,b_{d-1}$ to $B_{F}$. Each of the bits in the string represents the given answer of $B_{F}$ for a different question ruled by $k$. Thus, when  $B_{F}$ gets $k$ and the bit string $b$ he returns $G=b_{k}$ (figure \ref{fig:bRACBIBF}). This splitting in no way reduces the generality of Bob's behavior since the whole information processing part is done locally by $B_{I}$. $B_{F}$ only returns one of the values from a table provided by $B_{I}$.
\begin{figure}[h]
\centering
\includegraphics[width=\linewidth]{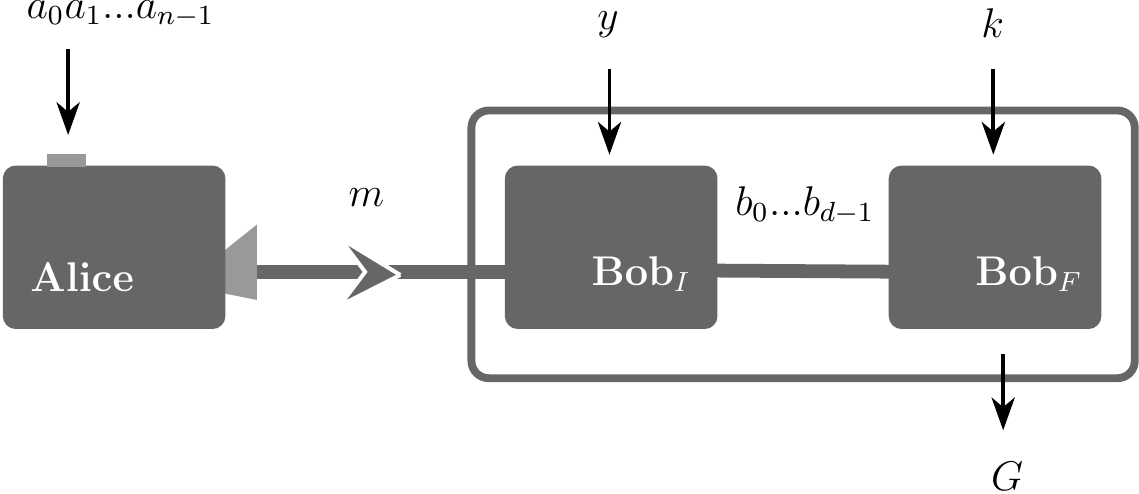}
\caption{Scheme of binary RAC. Bob is split into two parts $B_{I}$(initial Bob) and $B_{F}$(final Bob).}
\label{fig:bRACBIBF}
\end{figure}

Notice that before receiving Alice's message Bob knows nothing about the string $a$, so his entropy $H(a)=n \log d$ (we assume Alice's inputs are uniformly distributed). After receiving the message, Bob's entropy for each $a_{i}$ is reduced to $H_{i}^{m}=H(a_{i}|m)$.  These two entropies are related by information causality principle  \cite{Pawlowski2009}
\begin{align}\label{eq:ic}
 H(a)-\sum\limits_{i=0}^{n-1}H_{i}&\leq C,
\end{align}
where $H_{i}=\sum_{m=0}^{d-1} p(m)H_{i}^{m}$ is the averaged conditional Shannon entropy and $C$ is a capacity of a classical channel. Hence, from \eqref{eq:ic} one obtains the lower bound for $H_{i}$, which is determined by two established quantities: entropy $H(a)$ and the channel capacity $C$.

Besides the message $m$, $B_{I}$ receives the input $y$, which makes him interested in the particular dit $a_{y}$ from the string $a$. Let us introduce the following probability distribution $p_{j}=p(a_{y}=j|m,y)$, where $j\in\{0,\dots, d-1\}$, which represents $B_{I}$'s knowledge about dit $y$. Firstly, one sees that the entropy $H_{i=y}^{m}$ can be presented in terms of this probability distribution
\begin{align}\label{eq:entym}
H_{i=y}^{m}=-\sum\limits_{j=0}^{d-1}p_{j}\log p_{j}.
\end{align}
Secondly, one notices that depending on the payoff function, there exists a critical value of probability ($p_{crit}$) such that if $p_{j}>p_{crit}$ then sending $b_{j}=0$ leads to larger average payoff than $b_{j}=1$. We derive a formula for $p_{crit}$ in the following way. One knows that sending $b_{j}=0$ leads to the answer $G=0$ for $j=k$. This gives $T_{YES}$ points with probability $p_{j}$. For $b_{j}=1$ one gets 1 point with $1-p_{j}$. The first option is better if $T_{YES} p_{j} \geq 1-p_{j}$, so
\begin{align}\label{eq:pcrit}
p_j&\geq \frac{1}{T_{YES}+1}=p_{crit}.
\end{align}

Furthermore, let us analyze the average payoff $T=T(m,y)$ for a message set $m$, given encoding strategy $E$, the input $y$ and $T_{d}$ defined in \eqref{eq:avrBrac}
\begin{align}
T&=\frac{1}{T_{d}}\sum^{d-1}_{j=0}\bigg ( T_{YES}p(b_{j}=0|m,y)p(a_{y}=j|m,y)\\
&+p(b_{j}=1|m,y)p(a_{y}\neq j|m,y)\bigg). \nonumber 
\end{align}
We introduce a variable $x$ as the number of bits in the string $b$ for which the optimal strategy sets to $0$ for the probability distribution $p(b_{j}|m,y)$. In other words $x$ is the number of $p_{j}$s, which are greater than $p_{crit}$. Using $x$ one can rewrite the entropy $H_{i=y}^{m}$ \eqref{eq:entym} in the following way
\begin{align}\label{eq:entymbyx}
H_{i=y}^{m}=-\sum\limits_{j=0}^{x-1}p_{j}\log p_{j}-\sum\limits_{j=x}^{d-1}p_{j}\log p_{j}.
\end{align}
Additionally, without loss of generality, we may assume that $p_{j}$ are ordered in such way that $p_{j}\geq p_{j+1}$. Then the average payoff becomes
\begin{align}\label{eq:Tmypjx}
T&=\frac{1}{T_{d}}\Bigg[ \sum_{j=0}^{x-1}T_{YES}p_{j} +\sum_{j=x}^{d-1}(1-p_j) \Bigg].
\end{align}
Because the value of $T$ \eqref{eq:Tmypjx} depends only on the sums $\sum_{j=0}^{x-1}  p_{j}$ and $\sum_{j=x}^{d-1}  p_{j}$ and not on the individual elements of the sums, we can choose that all the elements in each sum are equal because this makes the entropy $H_{i=y}^{m}$ \eqref{eq:entymbyx} largest without changing $T$. In other words the probability distribution $p_{j}$ becomes a step function: the values of all $p_{j}$ for $j=\{0,\dots, x-1\}$ are uniform (denoted by $p$) and the values of the remaining $p_{j}$ for $j=\{x,\dots, d-1\}$ are uniform as well and, according to the normalization condition $\sum_{j} p_j=1$, they must be equal to $\frac{1-xp}{d-x}$. Obviously, we assume that the encoding strategy $E$ reaches $p>\frac{1}{d}$. Due to above assumptions  one can express $T$ as a function of $x$ and $p$
\begin{align}\label{eq:tyt}
T&=\frac{1}{T_{d}}\Bigg[x[T_{YES}\ p-(1-p)]+d-1 \Bigg]. 
\end{align}
The entropy \eqref{eq:entymbyx} (from now noted by $H^{x}$) can also be expressed by these parameters
\begin{align}\label{eq:hyt}
H^{x}&=-xp\log p-(1-xp)\log \frac{1-xp}{d-x}.
\end{align}
Imposing \eqref{eq:pcrit} we substitute $T_{YES}$ in \eqref{eq:tyt} and find
\begin{align}\label{eq:newprob}
p&=\frac{T+p_{crit}[d(T-1)-2T+x+1]}{x} .
\end{align}
One can further plug the above expression into \eqref{eq:hyt} to get the entropy $H^{x}$ as a function of $d,T,x$ and $p_{crit}$.
   
\subsection{Optimal \texorpdfstring{$x$}{x} for our case}
It has been shown in the section \ref{sec:OptClassStr} that the majority encoding is optimal in the standard RAC scenario, where Alice is allowed to send only one dit of information to Bob. To employ this result in the binary RAC protocol (in this case $B_{I}$ sends to $B_{F}$ a bit string $b_{0}\dots b_{d-1}$ with exactly one $0$ in the established position  and $1$s in the others) one must put the restriction that for any $T$,  probability $p$ for $x=1$ is always greater than any $p$ for  $x\neq 1$ \eqref{eq:newprob}. To make it, one must find a lower bound of  $p_{crit}$ such that the entropy $H^{x=1}$ is always greater than any entropy $H^{x\neq 1}$ for any given value of $T$ from the relevant range. Hence, in the beginning, we define a function $\Delta_{i}$ in the following way 

\begin{align}\label{eq:deltai}
\forall i\neq 1\  \Delta_{i}&=H^{x=1}-H^{x=i}\ .
\end{align}
Notice that the symmetry of the entropy $H^{x}=H^{d-x}$ for $x\in\{1,\dots,d-1\}$ causes that it is sufficient to check the condition \eqref{eq:deltai} only for $\Delta_{i}$, $i\in\{2,3,\dots,\lceil\frac{d}{2}\rceil\}$.

Let us outline the methodology of obtaining the minimum value of $p_{crit}$ for which $\Delta_i >0$. Clearly, $\Delta_i$ is a function of $d,T,p_{crit}$. We first find the range of $T$ in terms of $d$ and $p_{crit}$ within which $\Delta_i$ is well-defined. After that, we fix the value of $d$ and $p_{crit}$, and obtain the minimum value of $\Delta_i$ within the relevant range of $T$ for all $i$. If the minimum value of $\Delta_i$ is non-positive for some $i\in \{2,\dots,\lceil\frac{d}{2}\rceil\}$, we know such value of $p_{crit}$ is not suitable. We repeat the evaluation of $\Delta_i$ for another value of $p_{crit}$ increased by a small interval than before. Once we find that $\Delta_i$ is positive for all $i\in \{2,\dots,\lceil\frac{d}{2}\rceil\}$, we conclude that the taken value of $p_{crit}$ is approximately same as the desired value.

Firstly, for every $\Delta_{i}$ we must determine the range of $T$. The lower limit of the range is the value of $T$ for which $H^{x=1}$ is maximal. According to \eqref{eq:hyt} $H^{x=1}$ takes maximum  for $p=\frac{1}{d}$.  Putting it in \eqref{eq:newprob} gives an analytical expression for the lower limit
\begin{align}\label{eq:Tmax}
T_{0}&=\frac{1+(d-2)dp_{crit}}{d+(d-2)dp_{crit}}.
\end{align}
On the other hand, the upper limit of the range is the value of $T>T_{0}$ for which $H^{x}$ takes the bound. The bound is established by putting $xp=1$ in \eqref{eq:hyt}, so it strictly depends on $x$. Hence, setting  $p=\frac{1}{x}$  in \eqref{eq:newprob} gives
\begin{align}\label{eq:Tmin}
T_{1}^{x=i}&=\frac{1+p_{crit}(d-i-1)}{1+(d-2 )p_{crit}}.
\end{align}
Thus, for every $\Delta_{i}$ there is a different range $[T_{0},T_{1}^{x=i}]$.

We have found $p_{crit}$ numerically  using a method described by the following algorithm:
\begin{enumerate}
\item For a chosen dimension $d$, fix  $p_{crit}=\frac{1}{d}$ and $\epsilon_{p_{crit}}$  which is its numerical increase. 
\item Substitute $p_{crit}:=p_{crit}+\epsilon_{p_{crit}}$.\label{itm:pcritin}
\item Calculate $T_{0}$ from \eqref{eq:Tmax}.
\item Fix variable $i:=2$.
\item Calculate $T_{1}^{x=i}$ from \eqref{eq:Tmin}\label{itm:TminPoint}.
\item Calculate $\Delta_{i}$ for $T_{0}$, $T_{1}^{x=i}$ and find the minimal value of $\Delta_{i}$ in the range $[T_{0},T_{1}^{x=i}]$ (if the minimal value does not exist do not take it into account). If $\Delta_{i}\leq 0$ for at least one of these three (or two) points then go to the point \ref{itm:pcritin}. Otherwise, $i:=i+1$.\label{itm:deltapositive}
\item Check if $i \leq \lceil\frac{d}{2}\rceil$. If it is fulfilled  then go to the point \ref{itm:TminPoint}. Otherwise return $p_{crit}$. 
\end{enumerate}
Obviously, the accuracy of our method depends strictly on  $\epsilon_{p_{crit}}$.The smaller it is the more precise is the result.  Additionally, it is noteworthy that the criteria for optimal encoding is derived from $H_{i=y}^{m}$ \eqref{eq:entym} which is valid for all $y\in \{0,\dots,n-1\}$ and thus it is independent on $n$.
  
To illustrate the procedure described above we plot the dependence of $H$ on $T$ for some small $p_{crit}$ and different values of $x$ in figures \ref{fig:HTpcrit014} and \ref{fig:HTpcrit0185}. 
\begin{figure}[ht]
\centering
\includegraphics[width=\linewidth]{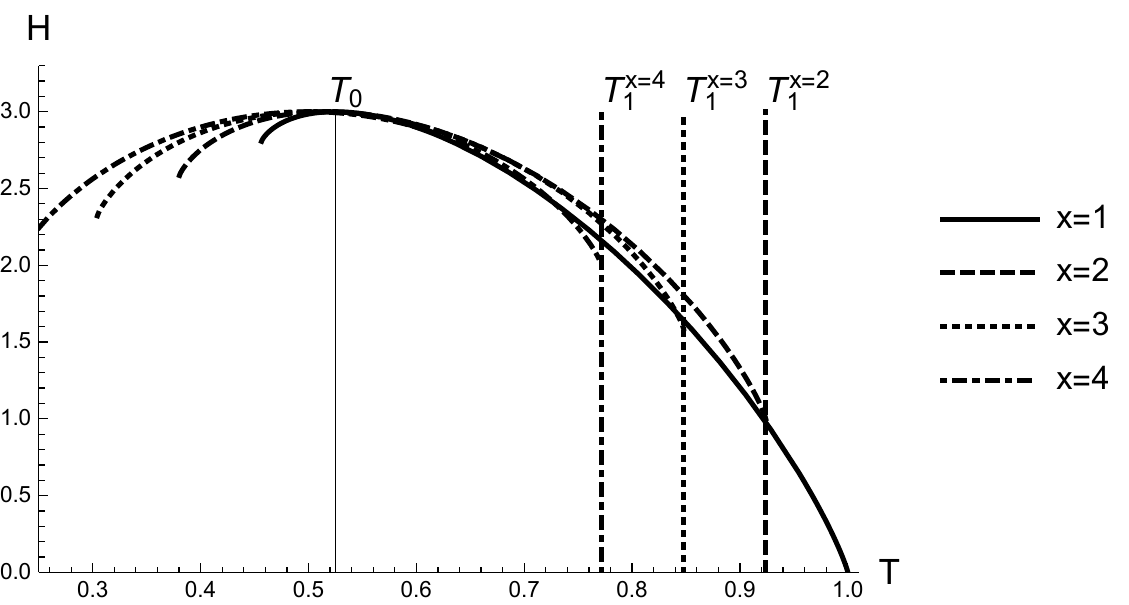}
\caption{Dependence of $H$ on $T$ for $d=8$, $x=1,2,3,4$ and $p_{crit}=0.14$. We note that the entropy for the strategy with $x=1$ is not always the largest in the established ranges of $T$. According to the numerical procedure, this is an example in which at the point \ref{itm:deltapositive} $\Delta_{i}\leq0$ and  our algorithm skips from the point \ref{itm:deltapositive} to the point \ref{itm:pcritin}. Vertical lines indicate limits of the ranges $[T_{0},T_{1}^{x=i}]$.}
\label{fig:HTpcrit014}
\end{figure}
\begin{figure}[ht]
\centering
\includegraphics[width=\linewidth]{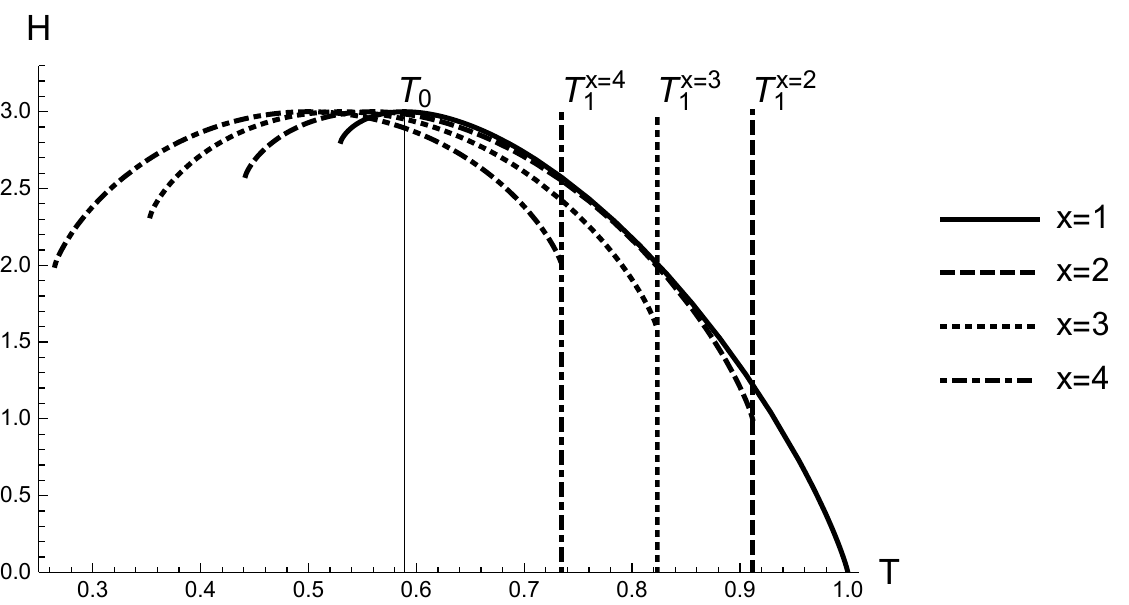}
\caption{Dependence of $H$ on $T$ for $d=8$, $x=1,2,3,4$ and $p_{crit}=0.18495$. Largest entropy is obtained with exactly one strategy for which $x=1$. According to our procedure, this is an example in which at the point \ref{itm:deltapositive} $\Delta_{i}>0$ for every $i\in\{2,3,\dots,\lceil\frac{d}{2}\rceil\}$ and our algorithm returns $p_{crit}$. Vertical lines indicate the limits of ranges $[T_{0},T_{1}^{x=i}]$.}
\label{fig:HTpcrit0185}
\end{figure}
Obtained values of $p_{crit}$ along with their corresponding $T_{YES}$ are shown in figure \ref{fig:plotdTyPcrit} and the values for some particular dimensions are mentioned in table \ref{tab:tabDimPcritTyes}. 
\begin{figure}[h]
\centering
\includegraphics[width=\linewidth]{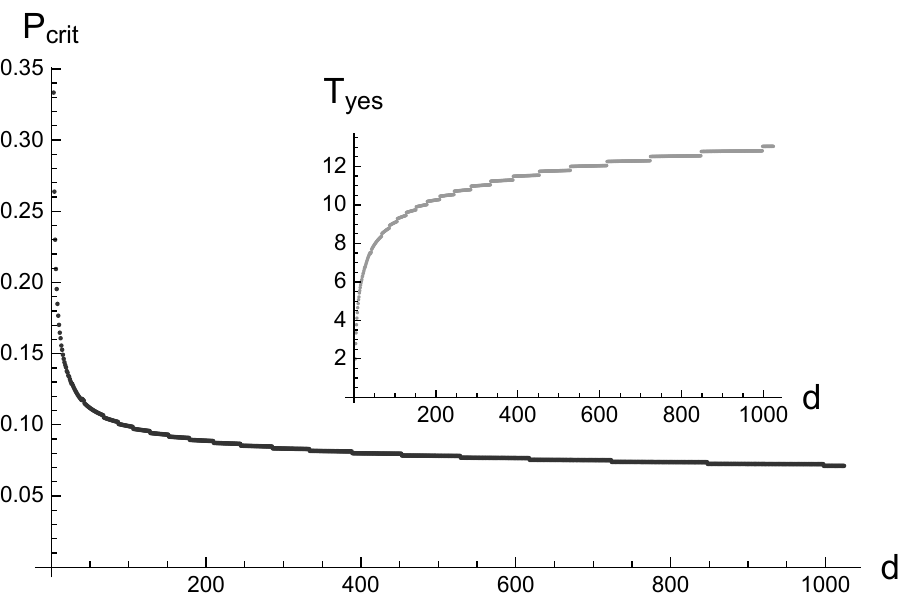}
\caption{Numerical calculation of values of minimal $p_{crit}$ and correspoding to it maximal $T_{YES}$ as a function of dimension with accuracy $\epsilon_{p_{crit}}=10^{-5}$.}\label{fig:plotdTyPcrit}
\end{figure}
\begin{table}[h!]
\begin{tabular}{|c|D{.}{.}{5}|D{.}{.}{5}|}
\hline
$d$ &\multicolumn{1}{c|}{$p_{crit}$}&\multicolumn{1}{c|}{$T_{YES}$}\\ 
\hline 
\bf{3} & 0.33340 & 1.99940\\ 
\hline 
\bf{8} & 0.18495 & 4.40687\\
\hline 
\bf{10} & 0.17021 & 4.87510\\
\hline 
\bf{50} & 0.11180 & 7.94454\\
\hline 
\bf{200} & 0.08885 & 10.25490\\
\hline 
\bf{700} & 0.07524 & 12.29080\\ 
\hline 
\bf{1000} & 0.07121 & 13.04300\\
\hline 
\end{tabular}
\caption{Values of minimal $p_{crit}$ and corresponding to it maximal $T_{YES}$ for chosen dimensions $d$.}
\label{tab:tabDimPcritTyes}
\end{table}
\subsection{The average classical and quantum payoff function for \texorpdfstring{$n=2$}{n=2} and arbitrary dimension}
Now we calculate the average classical and quantum payoff \eqref{eq:avrBrac} for binary RAC. Firstly, for a given dimension $d$, we must determine the value of $T_{YES}$ corresponding to $x=1$ as it was presented in the previous section. It follows that the optimal encoding strategy is sending the majority dit same as for the standard $d$RAC \eqref{eq:optde}. Further, it can be readily seen that given an encoding $E$ the optimal decoding is 
\begin{align}
G&=
\begin{cases}
    0 \ \text{if } \sum\limits_{a|a_y=k} \delta_{m,E(a)} \geq  \sum\limits_{a|a_y\neq k} \delta_{m,E(a)}\\
    1 \ \text{otherwise.} 
  \end{cases}
\end{align}
Therefore, in the case of majority encoding, Bob returns $G=0$ if the received message $m=k$, otherwise $1$. Given an input $a$ the total payoff over all possible $y,k$ is 
\begin{align}
T_{d }\tilde{n}&+(d-2)(n-\tilde{n}),
\end{align}
where we denote $\tilde{n}=\max\{n_{0},n_{1},\dots,n_{d-1}\}$. This is due to the fact that if $y$ is such that $n_{y}$ is the maximum, i.e.\ $a_{y}$ is the majority dit, then Bob gives the correct answer for all $k$, obtaining the maximum payoff $T_d$. Such event occurs $\tilde{n}$ times. In the other $(n-\tilde{n})$ cases Bob returns the correct answer only if $k\neq a_{y}$ and $k\neq E(a)$, obtaining $(d-2)$ payoff. Subsequently, the average payoff is given by
\begin{align}
\label{eq:avrPayoffGen}
\overline{T}_{B}^{c}&=\frac{1}{nd^{n}T_d}\sum \frac{n!}{n_0!n_1!\dots n_{d-1}!}\times\\
&\times[\tilde{n}(T_{YES}+1)+n(d-2)],\nonumber
\end{align}
where the summation is over all $\binom{n+d-1}{d-1}$ possible solutions of \eqref{eq:nways}. Imposing the expression of the average payoff of $d$RAC \eqref{eq:avrdRAC}, $\overline{T}_{B}^{c}$ simplifies to
\begin{align} 
\label{TcB}
\overline{T}_{B}^{c} &= \frac{(T_{YES}+1)\overline{T}_{S}^{c} + d-2}{T_{YES}+d-1}.
\end{align}
For $n=2$, one can find $\overline{T}_{S}^{c} = \frac{1}{2} + \frac{1}{2d}$, and substituting this in \eqref{TcB} leads to
\begin{align}
\label{payoffBRACn2Class}
\overline{T}_{B}^{c}&=\frac{1}{T_{d}}\bigg[\frac{T_{YES}+1+d(2d+T_{YES}-3)}{2d}\bigg]. 
\end{align}
Let us consider a quantum strategy based on the quantum $d$RAC presented in \cite{Tavakoli2015}. Alice codes her input $a_0a_1$ in $d$-dimensional quantum state as follows
\begin{align}\label{eq:codingOfAliceQuant}
\ket{\psi_{a_{0}a_{1}}}&=\frac{1}{N_{2,d}}\Bigg(\ket{a_{0}}+\frac{1}{\sqrt{d}}\sum_{j=0}^{d-1}\omega^{ja_{1}}\ket{a_{1}+j}\Bigg),
\end{align}
where $N_{2,d}=\sqrt{2+\frac{2}{\sqrt{d}}}$ is the normalization factor and $\omega=e^{2\pi i}$ is quantum Fourier transform factor. For the decoding Bob uses the following projective measurements $M_k^{y}$, depending on input $y,k$,
\begin{align}\label{eq:decProjSet}
M_k^{0}=\{P_{k}^{0},\mathbb{I}-P_{k}^{0}\},\ M_k^{1}=\{P_{k}^{1},\mathbb{I}-P_{k}^{1}\}.
\end{align}
Here, $P_{k}^{0}=\ket{k}\bra{k}$ 
and $P_{k}^{1}=\ket{\bar{k}}\bra{\bar{k}}$
taking $\ket{\bar{k}}=\frac{1}{\sqrt{d}}\sum_{k=0}^{d-1}\omega^{k\bar{k}}\ket{k}$ correspond to the outcome $G=0$. Simple calculations lead to the quantum average payoff
\begin{align}
\overline{T}_{B}^{q}&=\frac{1}{T_{d}}\bigg[\frac{T_{YES}+1+\sqrt{d}(2d+T_{YES}-3)}{2\sqrt{d}}\bigg].\label{eq:payoffBRACn2Quant}
\end{align}
The difference between \eqref{eq:payoffBRACn2Quant} and \eqref{payoffBRACn2Class} is given by 
\begin{align}\label{diffPayoffBRACn2}
\overline{T}_{B}^{q}-\overline{T}_{B}^{c}&=\frac{1}{T_{d}}\bigg[\frac{(T_{YES}+1)(\sqrt{d}-1)}{2d}\bigg],
\end{align}
which is always greater than zero for $d\geq2$. Thus, the binary version of RAC provides a device independent way to test arbitrary dimensional quantum system employing only binary outcome measurements.
\section{Summary}
The primary feature of this article is to present a DW applicable to test arbitrarily large quantum systems implementing only binary outcome measurements. We propose a new variant of RAC and take the average payoff of this communication task as the indicator of the dimension. We have provided the optimal classical bound for the  binary version of the generalized RAC. In contrast to the other quantum communication complexity problems in which the number of prepared states grows exponentially with dimension, the proposed DW requires $d^2$ different preparations and $2d$ measurements. 
In the future, it would be interesting to prove the optimality of the quantum strategy for binary RAC and look for more robust DWs retaining the aforementioned significant features. 
\begin{acknowledgments} 
We thank M{\'a}t{\'e} Farkas and Edgar A. Aguilar for helpful discussions and comments. We are also grateful to  
Edgar A. Aguilar for the revision of the manuscript.

This work was supported by FNP programme First TEAM (Grant No. First TEAM/2016-1/5, First TEAM/2017-4/31), NCN grants 2014/14/E/ST2/00020 and 2016/23/N/ST2/02817.
\end{acknowledgments}
\bibliography{bibliography}
\end{document}